\documentclass[letterpaper, 10 pt, conference]{IEEEconf}  % Comment this line out if you need a4paper
\usepackage[utf8]{inputenc}

\IEEEoverridecommandlockouts    
\usepackage{cite}
\usepackage{amsmath,amssymb,amsfonts}
\usepackage{graphicx}
\usepackage{bm}

\usepackage{algorithm}
\usepackage{algpseudocode}
\usepackage{amsmath,amssymb}
\usepackage{array}
\usepackage{makecell}

\usepackage{booktabs}
\usepackage{mathtools}
\usepackage{textcomp}
\usepackage{amssymb}
\usepackage{amsmath}
\usepackage{graphicx}
\usepackage{xcolor}
\usepackage{makecell}
\usepackage{comment}
\usepackage{todonotes}

\usepackage{cite}
\usepackage{graphicx}
\usepackage{subcaption}

\usepackage{xcolor}
\definecolor{light-blue}{rgb}{0.3,0.5,0.8}
\usepackage[colorlinks=true, linkcolor=cyan, citecolor=cyan, filecolor=magenta, urlcolor=cyan]{hyperref}
\usepackage{cleveref}

\title{\bf Closed-Form Characterization of Constrained Double-Integrator Optimal Control}

\author{Filippos N. Tzortzoglou, \IEEEmembership{Student Member, IEEE,} Logan E. Beaver \IEEEmembership{Member, IEEE,}\\ and Andreas A. Malikopoulos,
\IEEEmembership{Senior Member, IEEE}
\thanks{This research was supported in part by NSF under Grants CNS-2401007, CMMI-2348381, IIS-2415478, in part by MathWorks.}\thanks{Filippos N. Tzortzoglou is with the Civil and Environmental Engineering Department, Cornell University} \thanks{Logan Beaver is with the Mechanical \& Aerospace Engineering Department of Old Dominion University} \thanks{Andreas A. Malikopoulos is with the Applied Mathematics, Systems Engineering, Mechanical Engineering, Electrical \& Computer Engineering, and School of Civil \& Environmental Engineering, Cornell University, Ithaca, NY, USA. (email: \texttt{amaliko@cornell.edu})}}

\usepackage{amsthm}

\newtheorem{lemma}{Lemma} 
\newtheorem{theorem}{Theorem}

\newtheorem{problem}{Problem}
\newtheorem{remark}{Remark}

\newtheorem{assumption}{Assumption}

\begin{document}

\maketitle

\begin{abstract}
We consider the energy-optimal control problem for double-integrator systems subject to state and control constraints, with fixed terminal time and free terminal speed. When the constraints become active, the optimal trajectory consists of a combination of bang, unconstrained, and coast arcs, whose switching instants must be computed explicitly. In this paper, we derive closed-form expressions for the switching times of all admissible profiles, including both constrained and unconstrained arcs, reducing the computation in each case to explicit algebraic equations. In contrast to prior work, we classify all possible combinations of arcs, including special cases, and provide the specific conditions under which each case arises.
%at most a single quadratic equation, without the need to solve $4\times 4$ systems. 
Furthermore, we prove that when the initial unconstrained trajectory violates both speed and control constraints, the optimal solution follows a predetermined bang--affine--coast profile, enabling direct identification of the optimal trajectory without intermediate feasibility checks.
\end{abstract}  

\section{Introduction}
\label{sec:introduction}
Optimal control problems with state and control constraints arise naturally in many engineering systems, where physical limitations must be explicitly enforced \cite{bryson1969applied}. Properly accounting for these constraints is essential for both theoretical development and practical implementation. A representative example is an autonomous vehicle whose motion is governed by position and speed states, with acceleration/deceleration as the control input \cite{Malikopoulos2020}. In this setting, state constraints may enforce speed limits, while control constraints reflect bounds on acceleration and deceleration \cite{malikopoulos2019ACC}. Such formulations are fundamental in transportation and robotics applications, where safety and actuator limitations play a critical role.
% established a systematic framework for constrained optimal control through a wide range of engineering examples. Jacobson et al. ~\cite{jacobson1971} later derived sharper necessary conditions at the transition points between constrained and unconstrained arcs, refining the earlier results of~\cite{bryson1963optimal}. Bonnans and Hermant~\cite{bonnans2009} extended this analysis to vector-valued state constraints and controls, establishing second-order optimality conditions and characterizing the well-posedness of the associated shooting algorithm. Their framework, however, addresses general nonlinear systems and does not yield explicit closed-form expressions for the switching instants. Hartl et al.~\cite{hartl1995survey} later provided one of the most comprehensive surveys of maximum principles for problems with state constraints, while Evans~\cite{evans2024introduction} offered a rigorous modern introduction to optimal control theory.

The foundations of constrained optimal control trace back to the late 1960s~\cite{bryson1963optimal,bryson1969applied,jacobson1971,hartl1995survey,bonnans2009,evans2024introduction}, where the need to systematically handle state and control limitations first emerged. From an application perspective, early studies in aerospace engineering addressed optimal flight path problems with state-variable inequality constraints~\cite{speyer1969separate}. Since then, state- and control-constrained optimal control has been widely applied across several domains, including the guidance of autonomous underwater vehicles~\cite{chertovskih2020indirect}, locomotion of micro-robotic swimmers~\cite{wiezel2016using}, motion planning in robotics~\cite{beaver2023graph}, and efficiency improvements in transportation systems, particularly in the domain of connected and automated vehicles (CAVs)~\cite{Malikopoulos2020,meng2020eco,Meng2,xu2022general,bang2023optimal,Beaver2021DifferentialTime,tzortzoglou2026toward}.

In transportation applications, an important class of problems involves systems with double-integrator dynamics, where the objective is typically to minimize control effort for a CAV \cite{Malikopoulos2020,chalaki2021CSM}. In this setting, position and speed are the state variables, while acceleration serves as the control input. In \cite{locatelli2017}, a comprehensive study of optimal control for double-integrator systems was presented; however, it does not address the joint activation of both state and control constraints under an objective that minimizes control effort. In optimal control problems, when state and/or control constraints become active along the trajectory,  % with double-integrator dynamics, where the objective is to minimize control effort, 
the optimal solution is a %becomes 
piecewise function, consisting of a combination of constrained and unconstrained arcs. The switching instants between arcs must then be computed explicitly.

The authors in \cite{meng2020eco} developed a methodology to characterize optimal solutions for state- and control-constrained problems in the context of CAVs approaching signalized intersections. A key contribution of their work is the derivation of real-time analytical solutions, distinguishing their approach from prior methods based on numerical computation. For the fixed terminal time problem, however, they enumerate different cases (see Table II in \cite{meng2020eco}) without explicitly defining the resulting trajectory from the initial conditions. Also, identifying the switching instants between arcs in certain sub-cases requires solving coupled algebraic equations and, in some cases, cubic equations. The latter arise specifically when bang–affine–coast profiles occur. 

Later, in \cite{mahbub2020Automatica-2}, the authors addressed some of the limitations of \cite{meng2020eco}, focusing on CAVs operating at signal-free intersections. Specifically, they established a priori conditions for identifying constraint activation cases, and resulting control profiles. However, in both~\cite{meng2020eco} and~\cite{mahbub2020Automatica-2}, determining whether a single-constraint profile remains feasible requires sequential feasibility checks. In \cite{mahbub2020Automatica-2}, the trajectory is first classified as either bang–affine or affine–coast, depending on whether the unconstrained solution violates the control input or the speed limit, respectively. The resulting profile is then re-examined to check whether the introduced constrained arc activates the remaining constraint, in which case the solution escalates to a bang–affine–coast profile.

In this paper, we seek to contribute to this space. We first show that when the unconstrained trajectory violates both the state and control constraints, the optimal profile is directly of bang–affine–coast type, eliminating the need for sequential feasibility checks. In addition, in contrast to \cite{mahbub2020Automatica-2}, we provide conditions for all special cases that may arise, such as purely bang or bang–coast profiles. More importantly, we demonstrate that the switching instants between arcs for all admissible profiles reduce to solving at most a single quadratic equation, whose solution is an explicit function of the boundary data, with no intermediate constants of integration required.

The remainder of the paper is organized as follows. In Section II, we introduce the problem
formulation. In Section III, we classify the admissible control profiles and present the conditions
for their identification. In Section IV, we derive the closed-form analytical trajectories for each
profile.  Finally,
we offer concluding remarks and directions for future work in Section V.

\section{Problem Formulation}\label{sec:problem}

Consider an autonomous agent operating on a predefined path governed by double-integrator dynamics
\begin{align}\label{eq:dynamics}
    \dot{p}(t) &= v(t),  \\
    \dot{v}(t) &= u(t),\nonumber
\end{align}
where $p(t)\in\mathbb{R}$ is the position, $v(t)\in\mathbb{R}$ is the speed, and $u(t)\in\mathbb{R}$ is the acceleration (control input). The speed and control inputs are subject to the constraints
\begin{equation}\label{eq:constraints}
    v_{\min} \leq v(t) \leq v_{\max}, \qquad
    u_{\min} \leq u(t) \leq u_{\max},
\end{equation}
with $u_{\min} < 0 < u_{\max}$ and $v_{\min} < v_{\max}$.

Given an initial time $t^0$ with initial conditions $p(t^0)=p^0$, $v(t^0)=v^0$, and a given terminal position $p(T)=p^T=L$ at a fixed terminal time~$T$, we consider the following energy-minimizing optimal control problem. % without considering \eqref{eq:constraints} defined as
\begin{problem} \label{prb:ocp}
Determine the control trajectory $u(t)$ that minimizes,
\begin{align}
\label{eq:ocp}
\min_{u(\cdot)} \quad & \int_{t^0}^{T} u^2(t)\,dt \\
\text{s.t.} \quad 
& \eqref{eq:dynamics}, \eqref{eq:constraints},  \nonumber \\
& p(t^0)=p^0,\; v(t^0)=v^0, \nonumber \\
& p(T)=L.\nonumber 
\end{align}
\end{problem}
Next, we provide the optimal unconstrained trajectory for this problem. Note that although the process for identifying the unconstrained trajectory is well defined (see, e.g., \cite{Malikopoulos2020, mahbub2020Automatica-2}), here we prefer to give an alternative representation using only one coefficient $\alpha$ that will facilitate our exposition throughout the paper. In our exposition, we suppress dependence on $t$ for some variables for readability.

\begin{lemma}\label{unconstrained_solution} 
Without loss of generality, let $t^0=0$ and $p^0=0$. The optimal control input
$u(t)$, $t\in[0,T]$, for the unconstrained solution of Problem 1
is given by
\[
u(t)=\alpha (t-T),
\]
where
\[
\alpha=\frac{3(v^0T-L)}{T^3}.
\]
\end{lemma}

\begin{proof}
The Hamiltonian of Problem 1 is
\begin{equation}
H(u,p,v,\lambda_1,\lambda_2)=u^2+\lambda_1 v+\lambda_2 u,
\end{equation}
where $\lambda_1, \lambda_2$ are the time-varying costates.
For optimality, the costates must satisfy
\[
\dot{\lambda}_1=-\frac{\partial H}{\partial p}=0,
\qquad
\dot{\lambda}_2=-\frac{\partial H}{\partial v}=-\lambda_1.
\]
Thus,
\begin{equation}
\lambda_1=c_1,
\qquad
\lambda_2(t)=-c_1 t+c_2.
\end{equation}
From the stationarity condition,
\begin{equation}
\frac{\partial H}{\partial u}=2u+\lambda_2=0,
\end{equation}
we obtain
\begin{equation}
u(t)=-\frac{\lambda_2(t)}{2}
=\frac{c_1 t-c_2}{2}
=at+b,
\end{equation}
for some constants of integration $a,b$. Since the terminal speed is free, the transversality condition gives
$\lambda_2(T)=0.$
Hence, $u(T)=0,$ and therefore $u(t)=a(t-T)$. % for some constant $a$.

Let $\tau\in[0,T]$. Integrating $u(\tau)$ yields
\[
v(\tau)=v^0+\int_0^\tau a(s-T)\,ds
= v^0+\frac{a}{2}\tau^2-aT\tau.
\]
Integrating once more, we obtain
\begin{equation}
p(\tau)=0+\int_0^\tau v(s)\,ds
= 0 +v^0\tau+\frac{a}{6}\tau^3-\frac{aT}{2}\tau^2,
\end{equation}
which implies,
\[
L=v^0T+\frac{a}{6}T^3-\frac{a}{2}T^3
= v^0T-\frac{a}{3}T^3.
\]
Solving for $a$ gives
\[
a=\frac{3(v^0T-L)}{T^3}.
\]
Setting $\alpha=a$ completes the proof.
\end{proof}

We have identified the optimal control profile for the unconstrained version of Problem 1. However, our goal is to develop a framework that accounts for both the state and control constraints in \eqref{eq:constraints}. In \cite{mahbub2020Automatica-2}, it was shown that, when the final time $T$ is feasible, that is, reachable under the vehicle’s physical limits, the inclusion of state and control constraints leads to a finite number of possible control profile structures. In the next section, we discuss these profiles and their classification based on the vehicle’s initial state. We also provide conditions that lead to special control profiles not discussed in \cite{mahbub2020Automatica-2}.

\section{Optimal Control Profiles}\label{sec:profiles}
In this section, we discuss the different profiles our solution can exhibit based on the initial conditions of the autonomous agent. Before proceeding, note that, for a free terminal speed, the transversality condition yields an additional boundary condition %implies
%$$\lambda_2(T)=0 \;\Rightarrow\; 
$$u(T)=0,$$
which we derive in the sequel.
Moreover, by Lemma~\ref{unconstrained_solution}, the unconstrained optimal control is linear. Therefore, over the interval $[t^0,T]$, the acceleration cannot change sign.

\begin{assumption}\label{as:feasibility}
%Throughout the paper, we assume that 
The prescribed terminal time $T$ is \emph{feasible}; namely, there exists at least one admissible trajectory satisfying \eqref{eq:dynamics}--\eqref{eq:constraints} and the boundary condition. % that transfers the vehicle from $p^0$ to $p^T$ within the time interval $[t^0,T]$.
\end{assumption}

\begin{assumption}\label{as:symmetry}
%Throughout the paper, 
We restrict attention to the case in which the acceleration is linearly decreasing in time, equivalently, the initial acceleration is positive. As shown in \cite[Lemma 2]{mahbub2020Automatica-2}, this is equivalent to assuming that the initial speed satisfies
$$
v(t^0) < \frac{L}{T},
\qquad \text{where } L = p^T - p^0.
$$
The complementary case, in which the acceleration is linearly increasing in time, follows by symmetry and is therefore omitted.
\end{assumption}
\label{thm:six_profiles}
Based on the results in \cite{mahbub2020Automatica-2} and \cite{meng2020eco} and under Assumption~\ref{as:feasibility}, every optimal solution to Problem 1 with free terminal speed exhibits a control profile that is a concatenation of at most three arcs drawn from the set $\{\mathrm{Bang},\, \mathrm{Unconstrained},\, \mathrm{Coast}\}$, where
\begin{itemize}
    \item[\textbullet] a \emph{Bang} arc has $u = u_{\max}$ ,
    \item[\textbullet] an \emph{Unconstrained} arc has $u(t)$ affine in $t$, and
    \item[\textbullet] a \emph{Coast} arc has $u = 0$ with $v = v_{\max}.$ 
\end{itemize}
Specifically, the only admissible profiles are:
\begin{enumerate}
    \item[\textnormal{(P1)}] Unconstrained,
    \item[\textnormal{(P2)}] Bang\,$\to$\,Unconstrained,
    \item[\textnormal{(P3)}] Unconstrained\,$\to$\,Coast,
    \item[\textnormal{(P4)}] Bang\,$\to$\,Coast,
    \item[\textnormal{(P5)}] Bang\,$\to$\,Unconstrained\,$\to$\,Coast,
    \item[\textnormal{(P6)}] Bang.
\end{enumerate}
The relative control profiles are presented in Fig. \ref{fig:control_profiles}.

For profiles ending in unconstrained arcs (P1, P2), the transversality conditions imply,
\begin{equation}
    \lambda_2(T)=0 \implies u(T)=0.
\end{equation}
Similarly, for profiles ending on coast arcs (P3, P4, P5), the condition $u(T)=0$ is implied by the active constraint.
The bang profile is pathological and arises under a measure-zero set of conditions, which we describe in the following section.
Note that by definition, the bang profile $u=u_{\max}$ cannot change sign, thus it is consistent with prior discussion.

\begin{figure}
    \centering
    \includegraphics[width=1.01\linewidth]{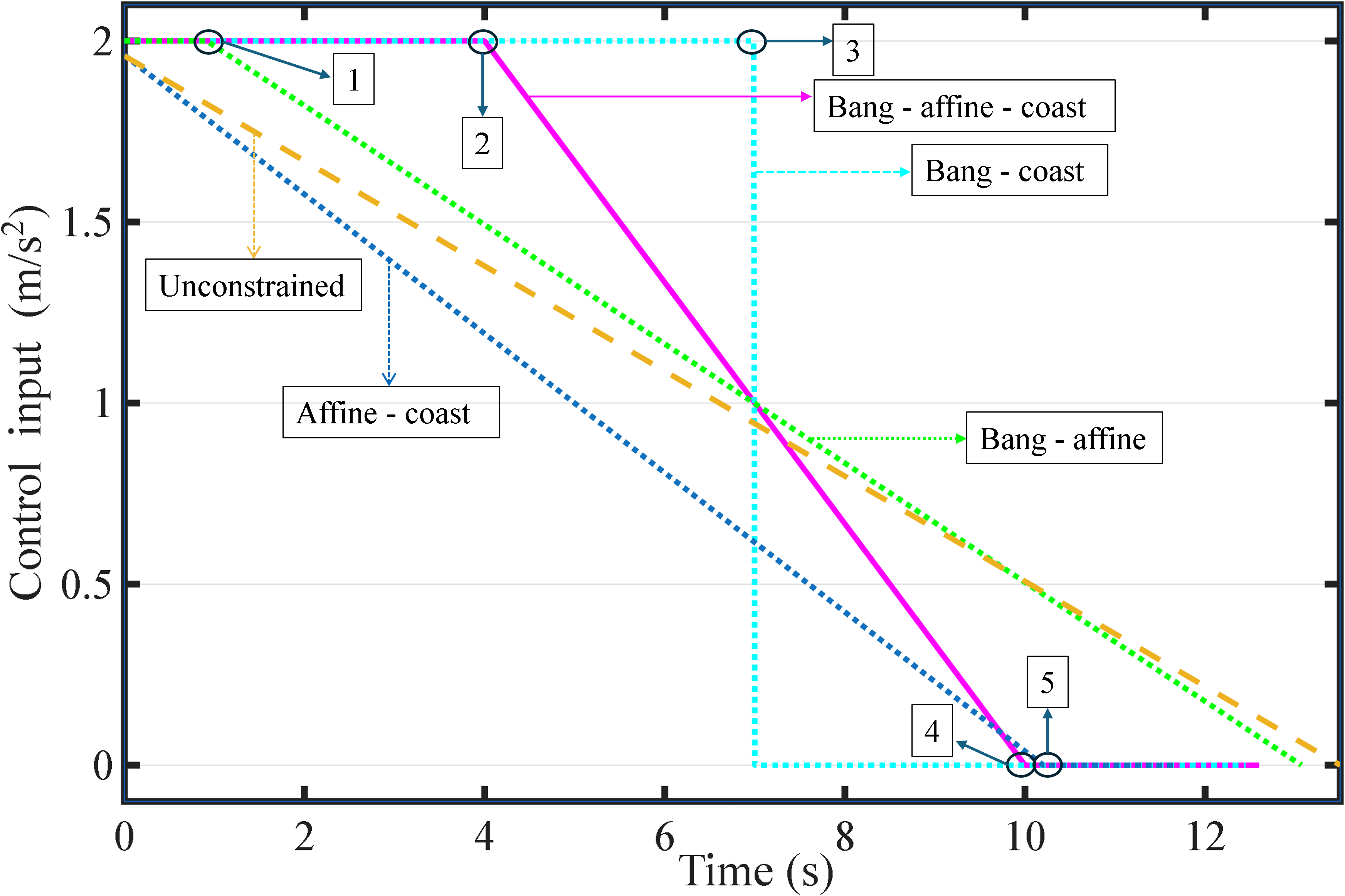}
    \caption{Different control profile combinations}
    \label{fig:control_profiles}
    \vspace{-12pt}
\end{figure}

\begin{table}[t]
\caption{Conditions for constraint activation in the accelerating case ($v^0 < L/T$, i.e., $a < 0$). A$=$affine, B--A $=$ bang-affine, A--C $=$ affine-coast, B--A--C $=$ bang-affine-coast. $\tau_c^*$: control constraint junction point;\; $\tau_s^*$: state constraint junction point. }
\label{tab:conditions}
\centering
\scriptsize
\renewcommand{\arraystretch}{1.0}
\setlength{\tabcolsep}{4pt}
\begin{tabular}{|c|c|l|}
\hline
\textbf{\#} & \textbf{Profile} & \textbf{Condition} \\ \hline
& & \\[-6pt]
1 & A &
$\displaystyle T > \frac{3L}{v^0 + 2v_{\max}} \;\text{\&}\; T > \frac{-3v^0 + \sqrt{9(v^0)^2 + 12\,u_{\max} L}}{2\,u_{\max}}$
\\[10pt] \hline
& & \\[-6pt]
2 & B--A &
$\begin{aligned}[t]
&\displaystyle\ T > \frac{3L}{v^0 + 2v_{\max}} \;\text{\&}\; T \leq \frac{-3v^0 + \sqrt{9(v^0)^2 + 12\,u_{\max} L}}{2\,u_{\max}} \; \\[6pt]
&\displaystyle\;\text{\&}\;\quad T < \tau_c^* - \frac{2\bigl(v_i(\tau_c^*) - v_{\max}\bigr)}{u_{\max}}
\end{aligned}$
\\[35pt] \hline
& & \\[-6pt]
3 & A--C &
$\begin{aligned}[t]
&\displaystyle T \leq \frac{3L}{v^0 + 2v_{\max}}\;\text{\&}\; T > \frac{-3v^0 + \sqrt{9(v^0)^2 + 12\,u_{\max} L}}{2\,u_{\max}} \\[6pt]
&\displaystyle\;\text{\&}\; \tau_s^* > \frac{-3v^0 + \sqrt{9(v^0)^2 + 12\,u_{\max}\bigl(p^*(\tau_s^*) - p(t^0)\bigr)}}{2\,u_{\max}}
\end{aligned}$\\[40pt] \hline
& & \\[-6pt]
4 & B--A--C &
$\begin{aligned}[t]
&\displaystyle T > \frac{3L}{v^0 + 2v_{\max}} \;\text{\&}\; T \leq \frac{-3v^0 + \sqrt{9(v^0)^2 + 12\,u_{\max} L}}{2\,u_{\max}} \\[6pt]
&\displaystyle\;\text{\&}\; T \geq \tau_c^* - \frac{2\bigl(v_i(\tau_c^*) - v_{\max}\bigr)}{u_{\max}}
\end{aligned}$
\\[35pt] \hline
& & \\[-6pt]
5 & B--A--C &
$\begin{aligned}[t]
&\displaystyle T \leq \frac{3L}{v^0 + 2v_{\max}} \;\text{\&}\; T > \frac{-3v^0 + \sqrt{9(v^0)^2 + 12\,u_{\max} L}}{2\,u_{\max}} \\[6pt]
&\displaystyle\;\text{\&}\; \tau_s^* \leq \frac{-3v^0 + \sqrt{9(v^0)^2 + 12\,u_{\max}\bigl(p^*(\tau_s^*) - p(t^0)\bigr)}}{2\,u_{\max}}
\end{aligned}$
\\[40pt] \hline
& & \\[-3pt]
6 & B--A--C &
$\displaystyle T \leq \frac{3L}{v^0 + 2v_{\max}} \;\text{\&}\; T \leq \frac{-3v^0 + \sqrt{9(v^0)^2 + 12\,u_{\max} L}}{2\,u_{\max}}$
\\[10pt] \hline
& & \\[-3pt]
7 & B &
$\begin{aligned}[t]
T = \frac{\sqrt{(v^0)^2 + 2u_{\max}L}}{u_{\max}} - \frac{v^0}{u_{\max}} \leq \frac{v_{\max} - v^0}{u_{\max}}
\end{aligned}$
\\[10pt] \hline
& & \\[-3pt]
8 & B-C &
$\begin{aligned}[t]
\tau_c = \tau_s &= \frac{v_{\max} - v^0}{u_{\max}} \\
3(v_{\max}-v^0)^2 &= 6 u_{\max}(T v_{\max} - L)
\end{aligned}$
\\[10pt] \hline
\end{tabular}
\vspace{4pt}
\raggedright
\end{table}

In \cite{mahbub2020Automatica-2}, conditions on the initial state of the agent were derived to determine the resulting control trajectory profile. Next, we review these conditions.

\subsection{Classification of control profiles}

Table~\ref{tab:conditions} summarizes the conditions under which different constraint activation cases arise in the accelerating case, i.e., when $v^0 < L/T$ and $a < 0$.
The conditions are organized in terms of two threshold quantities derived from the boundary conditions of Problem 1: the state constraint threshold, 
\begin{align}\label{state constraint threshold}
T \leq 3L/(v^0 + 2\,v_{\max}),\end{align} 
from Theorem~2(i) in~\cite{mahbub2020Automatica-2}, and the control constraint threshold, \begin{align} \label{control constraint threshold}
T \leq \bigl(-3v^0 + \sqrt{9(v^0)^2 + 12\,u_{\max}L}\bigr)/(2\,u_{\max}),\end{align}  from Theorem~3(i) in~\cite{mahbub2020Automatica-2}.

Note that the conditions in rows 6–8 of Table~\ref{tab:conditions} were not discussed in \cite{mahbub2020Automatica-2}. Next, we prove why these conditions result in the associated control profiles. We begin by formally stating the conditions under which the optimal trajectory given in Lemma \ref{unconstrained_solution} does not activate any constraint.

\begin{lemma}[Unconstrained Profile \cite{mahbub2020Automatica-2}] \label{lem:unconstrained}
Let $u_i^*(t) = a(t-T)$, $t \in [t^0,\, T]$, be the optimal control input for the unconstrained solution of~Problem 1 with $a < 0$.
If
\begin{equation}\label{eq:unc_cond}
    \begin{cases}
    T &> \frac{3L}{v^0 + 2\,v_{\max}} \\
    T &> \frac{-3v^0 + \sqrt{9(v^0)^2 + 12\,u_{\max}\,L}}{2\,u_{\max}},
    \end{cases}
\end{equation}
then none of the state and control constraints become active in $[t^0,\, T]$.
\end{lemma}

%\begin{proof}
%By Theorem~1 in~\cite{mahbub2020Automatica-2}, $a < 0$ implies that $v_{\min} - v_i(t) \leq 0$ and $u_{\min} - u_i(t) \leq 0$ cannot become active in $[t^0,\, T]$.
%The first condition in~\eqref{eq:unc_cond} is the negation of Theorem~2(i) in~\cite{mahbub2020Automatica-2}, which guarantees that $v_i(t) - v_{\max} \leq 0$ is not violated.
%The second condition is the negation of Theorem~3(i) in~\cite{mahbub2020Automatica-2}, which guarantees that $u_i(t) - u_{\max} \leq 0$ is not violated.
%Since no constraint becomes active, the unconstrained solution is optimal.
%\end{proof}

This result corresponds to the first row of Table~\ref{tab:conditions}, where the unconstrained (affine) solution satisfies all constraints, i.e., when $T$ exceeds both thresholds. A visualization of this control profile is shown in Fig.~\ref{fig:control_profiles} (yellow dotted line, labeled \textit{unconstrained}).

Next, we consider the impact of the constraint thresholds \eqref{control constraint threshold}, \eqref{state constraint threshold} on the trajectory.

\begin{lemma}[Constraint Activation \cite{mahbub2020Automatica-2}] \label{lma:threshold}
    If only the state threshold \eqref{state constraint threshold} is violated, then the trajectory will include a coast arc.
    If only the control threshold \eqref{control constraint threshold} is violated, then the trajectory will include a bang arc.
\end{lemma}

Lemma \ref{lma:threshold} applies to the second, third, fourth, and fifth rows of Table \ref{tab:conditions}.
In the second row, violating the control threshold \eqref{control constraint threshold} leads to a bang--affine (B--A) profile, while in the third row, violating the state threshold \eqref{state constraint threshold} leads to an affine--coast profile.
However, the threshold conditions are insufficient to determine whether these trajectories are bang--affine--coast (B--A--C).
This requires an additional condition, which we present next.

\begin{lemma}[B--A--C \cite{mahbub2020Automatica-2}] \label{lma:bac}
    Let $\tau_c^*$ and $\tau_s^*$ be the junction times for the bang--affine and affine--coast junctions.
    A problem satisfying Lemma \ref{lma:threshold} is bang--affine--coast if and only if,
    \begin{align}
        T &\geq \tau_c^\star - \frac{2\left(v_i(t_c^* - v_{\max}\right))}{u_{\max}}, \label{additional_check_bang_affine} \\
        \tau_s^* &\leq \frac{\sqrt{9(v^0)^2 + 12 u_{\max}\left( p^*(\tau_s^*) - p^0 \right)}}{2 u_{\max}}, \label{additional_check_affine_coast}
    \end{align}
    respectively, for each case.
\end{lemma}

%The second and third rows of Table~\ref{tab:conditions} correspond to cases in which only one constraint becomes active. In the first case, the control constraint threshold \eqref{control constraint threshold} is reached, while the state constraint threshold \eqref{state constraint threshold} is not, leading to a bang--affine (B--A) profile. In the second case, the state constraint threshold \eqref{state constraint threshold} is reached while the control constraint threshold \eqref{control constraint threshold} is not, leading to an affine--coast (A--C) profile. However, in both cases, an additional condition is required to determine the final control profile. In the (B--A) case we need to check additionally the condition
%\begin{align}\label{additional_check_bang_affine}
%    T < \tau_c^* - \frac{2\bigl(v_i(\tau_c^*) - v_{\max}\bigr)}{u_{\max}},
%\end{align}
%where $\tau_c^*$ is the switching time between bang and affine arcs (check box 1 in Fig. \ref{fig:control_profiles}). In the (A--C) case we need to additionally the condition 
%
%\begin{align}\label{additional_check_affine_coast}
%    \tau_s^* > \frac{-3v^0 + \sqrt{9(v^0)^2 + 12\,u_{\max}\bigl(p^*(\tau_s^*) - p(t^0)\bigr)}}{2\,u_{\max}}
%\end{align}
%where $\tau_s^*$ is the switching time between affine and coast arcs (check box 5 in Fig. \ref{fig:control_profiles}). The exact values of switching times $\tau_c^*$ and $\tau_s^*$ are defined in the next section.  

The conditions in Lemma \ref{lma:bac} determine whether the affine segment of the trajectory violates the state or control bounds, resulting in a bang--affine--coast solution in both cases.
%
%The reason for introducing conditions \eqref{additional_check_bang_affine} and \eqref{additional_check_affine_coast} in each case is that introducing a constrained arc changes the remaining affine portion of the trajectory, which may in turn activate the remaining constraint. In other words, a trajectory that initially appears to have a single active constraint may ultimately require a profile that considers activation of both constraints because of this coupling between the constrained and unconstrained arcs. This interdependence is captured by the conditions given in Theorems~4 and~5 in~\cite{mahbub2020Automatica-2}. 
%
Specifically, in Table~\ref{tab:conditions} the third inequality (\eqref{additional_check_bang_affine} and \eqref{additional_check_affine_coast}) in cases (B--A)  and (A--C)  ensures that no additional constraint becomes active after solving the corresponding single-constraint problem. If conditions \eqref{additional_check_bang_affine} and \eqref{additional_check_affine_coast} are not satisfied, the solution transitions directly to case (B--A--C). %For example, in Table \ref{tab:conditions} rows 2 and 4 (and similarly 3 and 5) differ only in the direction of the third inequality, which determines the optimal structure.

In \cite{mahbub2020Automatica-2}, a (B--A--C) profile is obtained only after sequentially examining the (B--A) and (A--C) cases. Specifically, if the conditions in rows 2 and 3 are not satisfied, the control profile escalates to (B--A--C). We next present a result showing that, when \eqref{control constraint threshold} and \eqref{state constraint threshold} hold simultaneously, the (B--A--C) profile arises directly, eliminating the need to check the (B--A) and (A--C) cases. This condition is captured in row 6 of Table \ref{tab:conditions}.

\begin{theorem}\label{thm:joint}
Let $u_i^*(t) = a(t-T)$, $t \in [t^0,\, T]$ be the optimal unconstrained solution of Problem 1, with $a < 0$.
If
\begin{equation}\label{eq:both_cond}
\begin{cases}
    T &\leq \frac{3L}{v^0 + 2\,v_{\max}}, \vspace{1em} \\
    T &\leq \frac{-3v^0 + \sqrt{9(v^0)^2 + 12\,u_{\max}\,L}}{2\,u_{\max}},
\end{cases}
\end{equation}
then the optimal solution of Problem~\ref{prb:ocp} activates both constraints. In particular, it is of bang–affine–coast type, with bang–coast arising as the degenerate boundary case. % activates both $v_i(t) - v_{\max} \leq 0$ and $u_i(t) - u_{\max} \leq 0$.
\end{theorem}

\begin{proof}
We prove that if \eqref{eq:both_cond} holds, neither the bang--affine nor the affine--coast profiles are feasible.
Therefore, both constraints must be
active in the optimal solution.
Let
\[
\beta := v_{\max}-v^0 > 0,
\qquad
\alpha := v_{\max}T-L,
\qquad
x := u_{\max}T,
\]
and, without loss of generality, let $t^0=0$. We first rewrite the two inequalities in \eqref{eq:both_cond} in forms that
will be used repeatedly. Re-arranging
\[
T \le \frac{3L}{v^0+2v_{\max}},
\]
yields
\begin{equation}\label{eq:thm1_alpha_bound} %\label{eq:thm1_first_rewrite}
%(v^0+2v_{\max})T \le 3L
%\;\Longrightarrow\;
3(v_{\max}T-L) \le (v_{\max}-v^0)T \implies 3\alpha \leq \beta T.%,
\end{equation}
%that is,
%\begin{equation}\label{eq:thm1_alpha_bound}
%3\alpha \le \beta T.
%\end{equation}
Multiplying by $u_{\max}$ gives
\begin{equation}\label{eq:thm1_ualpha_bound1}
3u_{\max}\alpha \le \beta x.
\end{equation}
Re-arranging the second inequality in \eqref{eq:both_cond} yields, %$ T \le \frac{-3v^0+\sqrt{9(v^0)^2+12u_{\max}L}}{2u_{\max}}$,
%and %using $x=u_{\max}T$, 
%we get
\begin{equation}
2x+3v^0 \le \sqrt{9(v^0)^2+12u_{\max}L}.    
\end{equation}
Squaring both sides and re-arranging yields $x^2+3v^0x \le 3u_{\max}L.$ Since $v_{\max}=v^0+\beta$, it follows that
\begin{align}
3u_{\max}\alpha
&= 3u_{\max}(v_{\max}T-L) \nonumber\\
&= 3(v^0+\beta)x - 3u_{\max}L \nonumber\\
&\le 3(v^0+\beta)x - (x^2+3v^0x) \nonumber\\
&= 3\beta x - x^2.
\label{eq:thm1_ualpha_bound2}
\end{align}

Next, we examine the two single-constraint profiles separately.

\smallskip
\noindent\textit{1) The affine--coast (A--C) profile cannot remain feasible on its own.} For the A--C profile, the switching time is
$
\tau_s^*=\frac{3\alpha}{\beta}
$
(cf. Lemma~6 in~\cite{mahbub2020Automatica-2}). The control on the affine arc
decreases linearly to zero at $\tau_s^*$, so its initial value is
$
u_i^*(0)=\frac{2\beta}{\tau_s^*}
=\frac{2\beta^2}{3\alpha}.
$
For the A--C profile to be a \emph{single}-constraint solution, we must have
$u_i^*(0)<u_{\max}$; equivalently,
\begin{equation}\label{eq:AC_feasibility}
3u_{\max}\alpha > 2\beta^2.
\end{equation}
We show that \eqref{eq:AC_feasibility} is impossible under \eqref{eq:both_cond}. If $x\le 2\beta$, then from \eqref{eq:thm1_ualpha_bound1},
\[
3u_{\max}\alpha \le \beta x \le 2\beta^2.
\]
If instead $x>2\beta$, then from \eqref{eq:thm1_ualpha_bound2},
\[
3u_{\max}\alpha \le 3\beta x - x^2.
\]
Moreover, for $x\ge 2\beta$,
\[
(x-\beta)(x-2\beta)\ge 0
\;\Longrightarrow\;
x^2-3\beta x+2\beta^2\ge 0,
\]
hence
\[
3\beta x - x^2 \le 2\beta^2.
\]
Therefore, in both cases,
\[
3u_{\max}\alpha \le 2\beta^2,
\]
which implies
\[
u_i^*(0)=\frac{2\beta^2}{3\alpha}\ge u_{\max}.
\]
Thus, the control constraint is active as well, and the A--C profile cannot be
a single state-constrained solution.

\smallskip
\noindent\textit{2) The bang--affine (B--A) profile cannot remain feasible on its own.} For the B--A profile, the switching time is
$
\tau_c^*
=
T-\sqrt{3}\sqrt{T^2-\frac{2(L-v^0T)}{u_{\max}}}
$
(cf. Lemma~7 in~\cite{mahbub2020Automatica-2}), and the terminal speed is
$
v_i^*(T)=v^0+\frac{u_{\max}}{2}(T+\tau_c^*).
$
Substituting the expression for $\tau_c^*$ and using $x=u_{\max}T$, we obtain
\begin{equation}\label{eq:vT_BA}
v_i^*(T)
=
v^0 + x - \frac{\sqrt{3}}{2}
\sqrt{x^2-2u_{\max}(L-v^0T)}.
\end{equation}
From the first inequality in \eqref{eq:both_cond},
\[
(v^0+2v_{\max})T \le 3L,
\]
so
\[
3(L-v^0T)\ge 2(v_{\max}-v^0)T = 2\beta T,
\]
that is,
\begin{equation}\label{eq:Lminusv0T_bound}
L-v^0T \ge \frac{2\beta T}{3}
= \frac{2\beta x}{3u_{\max}}.
\end{equation}
Substituting \eqref{eq:Lminusv0T_bound} into \eqref{eq:vT_BA} gives
\[
v_i^*(T)
\ge
v^0 + x - \frac{\sqrt{3}}{2}
\sqrt{x^2-\frac{4\beta x}{3}}.
\]
Now observe that
\[
(x-\beta)^2 - \frac{3}{4}\left(x^2-\frac{4\beta x}{3}\right)
=
\frac{(x-2\beta)^2}{4}\ge 0.
\]
Hence,
\[
x-\beta \ge \frac{\sqrt{3}}{2}\sqrt{x^2-\frac{4\beta x}{3}},
\]
and therefore
\[
v_i^*(T)\ge v^0+\beta = v_{\max}.
\]
Thus, the speed constraint is active as well, and the B--A profile cannot be a
single control-constrained solution.

\smallskip
Since neither single-constraint profile can remain feasible under
\eqref{eq:both_cond}, both $u_i(t)-u_{\max}\le 0$ and
$v_i(t)-v_{\max}\le 0$ must be active in the optimal solution.
\end{proof}

Theorem~\ref{thm:joint} allows direct identification of the B--A--C profile from the boundary conditions of Problem 1, eliminating the need for the intermediate interdependence check via Theorems~4 and~5 in~\cite{mahbub2020Automatica-2} when both conditions in~\eqref{eq:both_cond} hold.

For completeness, we also present the conditions for the bang-and bang-coast solutions, given in rows 7 and 8 of Table \ref{tab:conditions}.
These are pathological cases, which only occur for a measure-zero set of initial conditions.
We present the bang conditions next, and provide bang-coast in the following section.

\begin{lemma} \label{lma:bang}
The bang arc is optimal if and only if,
\begin{equation}
        T =  \frac{\sqrt{v_0^2 + 2u_{\max}L}}{u_{\max}} - \frac{v_0}{u_{\max}} \leq \frac{v_{\max}-v_0}{u_{\max}},
\end{equation}
which is a measure-zero set.
\end{lemma}

\begin{proof}
Substituting the boundary conditions of Problem 1 into constant-acceleration kinematics yields,    
\begin{align*}
    L &= v_0 T+\frac{1}{2}u_{\max} T^2,\\
    v_f &=v_0+u T \leq v_{\max}.
\end{align*}
Solving both equations for $T$ completes the proof.
\end{proof}

Lemma \ref{lma:bang} is provided for completeness, as the bang case only occurs at a single point on the boundary of the admissible initial conditions set under Assumption \ref{as:feasibility}.

Finally, the bang–coast case (the last row in Table \ref{tab:conditions}) is a degenerate boundary case of the bang–affine–coast profile. In particular, it arises when the two switching times of the bang–affine–coast solution coincide, so that the intermediate affine arc collapses to zero duration. For this reason, we do not treat bang–coast separately at this stage. Instead, in the next section, when deriving the switching times for the bang–affine–coast profile, we show how the bang–coast structure follows directly when these switching times are identical.
%
%Building upon \cite{mahbub2020Automatica-2}, we have established the conditions that determine the resulting control profile.

\section{Analytical constrained trajectories}
In this section, we propose the analytical trajectories per profile that are implemented by the agent after the classification presented in Table \ref{tab:conditions}.

\subsection{Affine-coast profile}
Next, we showcase the exact affine-coast (A-C) control profile along with the switching time $\tau_s^*$.
\begin{lemma}\label{lem:AC_junction}
Consider the optimal control Problem~1, and suppose that only $v_i(t) - v_{\max} \leq 0$ becomes active.
Then the optimal trajectory consists of two arcs: an unconstrained arc on $[t^0,\, \tau_s^*]$ with control input
\begin{equation}\label{eq:ac_control}
    u_i^*(t) = u_0\!\left(1 - \frac{t - t^0}{\tau_s^* - t^0}\right), \quad t \in [t^0,\, \tau_s^*],
\end{equation}
where $u_0 := u_i^*(t^0)$, followed by a state-constrained arc $v_i^*(t) = v_{\max}$, $u_i^*(t) = 0$ on $[\tau_s^*,\, T]$.
The junction point is given explicitly by
\begin{equation}\label{eq:ac_tau}
    \tau_s^* = \frac{3(v_{\max}\, T - L)}{v_{\max} - v^0},
\end{equation}
where $L := p(T)$ for $p(t^0)=0$.
\end{lemma}

\begin{proof}
Without loss of generality, let $t^0 = 0$ and $p(t^0)=0$  and $\beta := v_{\max} - v^0$. We present the proof in three steps.

\smallskip
\noindent\textit{Step~1: Control and speed on the unconstrained arc.}
From the necessary conditions~\cite{mahbub2020Automatica-2}, the optimal control on the unconstrained arc is linear in time with $u_i^*(\tau_s^*) = 0$, yielding~\eqref{eq:ac_control}.
Integrating, the speed on $[0,\,\tau_s^*]$ yields
\begin{equation}\label{eq:ac_speed}
    v_i(t) = v^0 + u_0\!\left(t - \frac{t^2}{2\tau_s^*}\right).
\end{equation}
Evaluating at $t = \tau_s^*$ and imposing the continuity condition $v_i(\tau_s^*) = v_{\max}$, we obtain
\begin{equation}\label{eq:ac_u0}
    u_0 = \frac{2\beta}{\tau_s^*}.
\end{equation}

\smallskip
\noindent\textit{Step~2: Position at the junction point.}
Integrating the speed profile~\eqref{eq:ac_speed} over $[0,\,\tau_s^*]$ and substituting~\eqref{eq:ac_u0}, the position at the junction point is
\begin{equation}\label{eq:ac_pos}
    p(\tau_s^*) = v^0\,\tau_s^* + \frac{u_0 (\tau_s^*)^2}{3} = v^0\,\tau_s^* + \frac{2\beta\,\tau_s^*}{3}.
\end{equation}

\smallskip
\noindent\textit{Step~3: Applying the terminal position constraint.}
On the state-constrained arc $[\tau_s^*,\, T]$, the speed is constant at $v_{\max}$, so
\begin{equation}\label{eq:ac_terminal}
    p(\tau_s^*) + v_{\max}(T - \tau_s^*) = L.
\end{equation}
Substituting~\eqref{eq:ac_pos} into~\eqref{eq:ac_terminal} yields
\begin{equation*}
    v^0\,\tau_s^* + \frac{2\beta}{3}\,\tau_s^* + v_{\max}\,T - v_{\max}\,\tau_s^* = L.
\end{equation*}
Collecting the coefficient of $\tau_s^*$ and noting that $v^0 + \frac{2\beta}{3} - v_{\max} = -\frac{\beta}{3}$, simplifies to
$-\frac{\beta}{3}\,\tau_s^* + v_{\max}\,T = L.$
Solving for $\tau_s^*$ yields~\eqref{eq:ac_tau}.
\end{proof}

% \begin{remark}
% The expression~\eqref{eq:ac_tau} is consistent with Lemma~6 in~\cite{mahbub2020Automatica-2}, which states $\tau_s^* = 3(p(T) - v_{\max}\, T)/(v^0 - v_{\max})$.
% Noting that $(v_{\max}\, T - L)/(v_{\max} - v^0) = (L - v_{\max}\, T)/(v^0 - v_{\max})$, the two expressions are identical.
% However, the derivation above is considerably more direct.
% In~\cite{mahbub2020Automatica-2} (Appendix~A), the junction point is obtained by solving a system of three coupled equations~(A.2a)--(A.2c) involving the constants of integration $a^{(1)}$, $b^{(1)}$ of the unconstrained arc, and eliminating them through successive substitution.
% Here, exploiting the speed continuity condition $v_i(\tau_s^*) = v_{\max}$ to express $u_0$ in terms of $\tau_s^*$ (Step~1) reduces the problem to a single linear equation in $\tau_s^*$ (Step~3), yielding the closed-form expression without requiring any intermediate constants of integration.
% \end{remark}

\subsection{Bang-affine profile}

Next, we showcase the bang-affine profile. 

\begin{lemma}\label{lem:BA_junction}
Consider the optimal control Problem~1, and suppose that only $u_i(t) - u_{\max} \leq 0$ becomes active.
Then the optimal trajectory consists of two arcs: a control-constrained arc $u_i^*(t) = u_{\max}$ on $[t^0,\, \tau_c^*]$, followed by an unconstrained arc on $[\tau_c^*,\, T]$ with control input
\begin{equation}\label{eq:ba_control}
    u_i^*(t) = \frac{u_{\max}}{T - \tau_c^*}\,(T - t), \quad t \in [\tau_c^*,\, T].
\end{equation}
The junction point is given explicitly by
\begin{equation}\label{eq:ba_tau}
    \tau_c^* = T - \sqrt{3}\,\sqrt{(T)^2 - \frac{2(L - v^0\, T)}{u_{\max}}},
\end{equation}
where $L := p(T)$ for $p(t^0)=0$.
\end{lemma}

\begin{proof}
Without loss of generality, let $t^0 = 0$ and $p(t^0) = 0$.

\smallskip
\noindent\textit{Step~1: Control on the unconstrained arc.}
From the necessary conditions~\cite{mahbub2020Automatica-2}, the optimal control on the unconstrained arc is linear in time with $u_i^*(\tau_c^*) = u_{\max}$ (continuity at the junction point) and $u_i^*(T) = 0$ (terminal condition from Lemma~1 in~\cite{mahbub2020Automatica-2}).
Interpolating between these two points yields~\eqref{eq:ba_control}.

\smallskip
\noindent\textit{Step~2: Speed and position at the terminal time.}
On the control-constrained arc $[0,\,\tau_c^*]$, the speed and position are
\begin{equation}\label{eq:ba_bang}
    v_i(\tau_c^*) = v^0 + u_{\max}\,\tau_c^*, \qquad p(\tau_c^*) = v^0\,\tau_c^* + \tfrac{1}{2}\,u_{\max}\,(\tau_c^*)^2.
\end{equation}
On the unconstrained arc $[\tau_c^*,\, T]$, integrating~\eqref{eq:ba_control} yields the terminal speed
\begin{equation}\label{eq:ba_vT}
    v_i(T) = v^0 + \frac{u_{\max}(T + \tau_c^*)}{2},
\end{equation}
and, using the position propagation on each arc together with~\eqref{eq:ba_bang}, the terminal position reduces to
\begin{equation}\label{eq:ba_xT}
    p(T) = v^0\,T + u_{\max}\!\left(\frac{T^2}{3} + \frac{T\,\tau_c^*}{3} - \frac{(\tau_c^*)^2}{6}\right).
\end{equation}

\smallskip
\noindent\textit{Step~3: Solving for the junction point.}
Imposing the terminal position constraint $p(T) = L$ in~\eqref{eq:ba_xT} yields
\begin{equation}\label{eq:ba_pos_constraint}
    L - v^0\,T = u_{\max}\!\left(\frac{T^2}{3} + \frac{T\,\tau_c^*}{3} - \frac{(\tau_c^*)^2}{6}\right).
\end{equation}
Multiplying both sides by $6/u_{\max}$ and rearranging, we obtain the quadratic equation
\begin{equation}\label{eq:ba_quadratic}
    (\tau_c^*)^2 - 2T\,\tau_c^* - 2T^2 + \frac{6(L - v^0\,T)}{u_{\max}} = 0.
\end{equation}
Applying the quadratic formula, the two roots are
\begin{align}\label{eq:ba_roots}
    \tau_c^{(\pm)} &= T \pm \sqrt{3T^2 - \frac{6(L - v^0\,T)}{u_{\max}}} \nonumber \\ &= T \pm \sqrt{3}\,\sqrt{T^2 - \frac{2(L - v^0\,T)}{u_{\max}}}.
\end{align}
The root $\tau_c^{(+)} = T + \sqrt{3}\,\sqrt{\cdot} > T$ violates $\tau_c^* < T$ and is therefore inadmissible.
The admissible junction point is $\tau_c^{(-)}$, which yields~\eqref{eq:ba_tau}.
\end{proof}

\subsection{Case where the control input is violated at the beginning (bang-affine-coast)}
Next, we present the case of the bang-affine-coast control profile. Here, for the identification of the switching times $\tau_s^*$ and $\tau_c^*$, we end up in a significantly simpler formulation in comparison to \cite{mahbub2020Automatica-2}.

\begin{lemma}\label{lem:BAC_junctions}
Consider the optimal control Problem~1, and suppose that both $u_i(t) - u_{\max} \leq 0$ and $v_i(t) - v_{\max} \leq 0$ become active.
Then the optimal trajectory consists of three arcs: a control-constrained arc $u_i^*(t) = u_{\max}$ on $[t^0,\, \tau_c^*]$, an unconstrained arc $u_i^*(t) = a^{(2)} t + b^{(2)}$ on $[\tau_c^*,\, \tau_s^*]$, and a state-constrained arc $v_i^*(t) = v_{\max}$, $u_i^*(t) = 0$ on $[\tau_s^*,\, T]$.
The junction points are given explicitly by
\begin{align}
    \tau_c^* &= \frac{\beta - \sqrt{\Psi}}{u_{\max}}, \label{eq:tau_c}\\[4pt]
    \tau_s^* &= \frac{\beta + \sqrt{\Psi}}{u_{\max}}, \label{eq:tau_s}
\end{align}
where $\beta := v_{\max} - v^0$ and $\Psi := 6\,u_{\max}(T\, v_{\max} - L) - 3\beta^2$, and $L := p(T)$ for $p(t^0)=0$.
\end{lemma}

\begin{proof}
Without loss of generality, let $t^0 = 0$ and $p(t^0) = 0$.

\smallskip
\noindent\textit{Step~1: Relating junction points via speed continuity.}
On the control-constrained arc $[0,\,\tau_c^*]$, the speed evolves as $v_i(t) = v^0 + u_{\max}\,t$, yielding
\begin{equation}\label{eq:v_tau1}
    v_i(\tau_c^*) = v^0 + u_{\max}\,\tau_c^*.
\end{equation}
On the unconstrained arc $[\tau_c^*,\,\tau_s^*]$, the optimal control is linear in time with $u_i^*(\tau_c^*) = u_{\max}$ and $u_i^*(\tau_s^*) = 0$ (from the necessary conditions~\cite{mahbub2020Automatica-2}), which yields
\begin{equation}\label{eq:u_affine}
    u_i^*(t) = \frac{-u_{\max}}{\tau_s^* - \tau_c^*}\,t + \frac{u_{\max}\,\tau_s^*}{\tau_s^* - \tau_c^*}, \quad t \in [\tau_c^*,\, \tau_s^*].
\end{equation}
Integrating~\eqref{eq:u_affine} over $[\tau_c^*,\,\tau_s^*]$ and using~\eqref{eq:v_tau1}, the speed at $\tau_s^*$ is
\begin{equation}\label{eq:v_tau2}
    v_i(\tau_s^*) = v^0 + \frac{u_{\max}(\tau_c^* + \tau_s^*)}{2}.
\end{equation}
Since the state-constrained arc requires $v_i(\tau_s^*) = v_{\max}$, we obtain from~\eqref{eq:v_tau2}
\begin{equation}\label{eq:sum}
    \tau_c^* + \tau_s^* = \frac{2\beta}{u_{\max}},
\end{equation}
where $\beta := v_{\max} - v^0$.
Therefore,
\begin{equation}\label{eq:tau1_from_tau2}
    \tau_c^* = \frac{2\beta}{u_{\max}} - \tau_s^*.
\end{equation}

\smallskip
\noindent\textit{Step~2: Deriving the position at the state constraint junction.}
The position at $\tau_c^*$ from the control-constrained arc is $p(\tau_c^*) = v^0\,\tau_c^* + \frac{1}{2}u_{\max}(\tau_c^*)^2$.
Integrating the speed profile on the unconstrained arc and simplifying using~\eqref{eq:tau1_from_tau2}, the position at $\tau_s^*$ reduces to
\begin{equation}\label{eq:p_tau2}
    p(\tau_s^*) = \frac{4\beta}{3}\,\tau_s^* - \frac{2\beta^2}{3\,u_{\max}} - \frac{u_{\max}}{6}(\tau_s^*)^2 + v^0\,\tau_s^*.
\end{equation}

\smallskip
\noindent\textit{Step~3: Applying the terminal position constraint.}
On the state-constrained arc $[\tau_s^*,\, T]$, the speed is constant at $v_{\max}$, so
\begin{equation}\label{eq:terminal}
    p(\tau_s^*) + v_{\max}(T - \tau_s^*) = L.
\end{equation}
Substituting~\eqref{eq:p_tau2} into~\eqref{eq:terminal} and using $v^0 + \frac{4\beta}{3} - v_{\max} = \frac{\beta}{3}$, we obtain the quadratic equation
\begin{equation}\label{eq:quadratic}
    \frac{u_{\max}}{2}(\tau_s^*)^2 - \beta\,\tau_s^* + \frac{2\beta^2}{u_{\max}} - 3(T\,v_{\max} - L) = 0.
\end{equation}

\smallskip
\noindent\textit{Step~4: Solving for $\tau_s^*$ and $\tau_c^*$.}
Applying the quadratic formula to~\eqref{eq:quadratic} with $A = u_{\max}/2$, $B = -\beta$, and $C = 2\beta^2/u_{\max} - 3(T\,v_{\max} - L)$, the discriminant is
\begin{align}
    \Delta &= \beta^2 - 2\,u_{\max}\!\left(\frac{2\beta^2}{u_{\max}} - 3(T\,v_{\max} - L)\right)  \nonumber \\
    &=6\,u_{\max}(T\,v_{\max} - L) - 3\beta^2 \nonumber \\
    &= \Psi.
\end{align}
The two roots are
\begin{equation}\label{eq:roots}
    \tau_s^{(\pm)} = \frac{\beta \pm \sqrt{\Psi}}{u_{\max}}.
\end{equation}
For the B--A--C profile to be well-posed, we require $0 < \tau_c^* < \tau_s^* \leq T$.
From~\eqref{eq:sum}, $\tau_c^* < \tau_s^*$ is equivalent to $\tau_s^* > \beta/u_{\max}$.
The smaller root $\tau_s^{(-)} = (\beta - \sqrt{\Psi})/u_{\max} < \beta/u_{\max}$ violates this requirement.
The larger root $\tau_s^{(+)} = (\beta + \sqrt{\Psi})/u_{\max} > \beta/u_{\max}$ satisfies it.
Therefore, $\tau_s^* = (\beta + \sqrt{\Psi})/u_{\max}$, and from~\eqref{eq:tau1_from_tau2}, $\tau_c^* = (\beta - \sqrt{\Psi})/u_{\max}$.
\end{proof}

\begin{remark}
The boundary case $\Psi = 0$ corresponds to $\tau_c^* = \tau_s^* = \beta/u_{\max}$, where the unconstrained arc collapses and the profile degenerates into a bang-coast structure.
\end{remark}

\begin{remark}\label{rem:comparison}
The junction points in Lemma~\ref{lem:BAC_junctions} provide a significant simplification over the expressions derived in Lemma~8 and Appendix~C of~\cite{mahbub2020Automatica-2}.
Therein, the junction points are obtained through the intermediate constants $a^{(2)} = -u_{\max}^2\sqrt{-1/\phi}$ and $b^{(2)} = u_{\max}(-2v^0\sqrt{-1/\phi} + 2v_{\max}\sqrt{-1/\phi} + 1)/2$, where $\phi$ is a function of five parameters.
Here, the junction points $\tau_c^*$ and $\tau_s^*$ are then recovered indirectly from these constants.
The resulting closed-form expressions~\eqref{eq:tau_c}--\eqref{eq:tau_s} are explicit functions of the boundary data $(v^0, L, T, u_{\max}, v_{\max})$, are symmetric in structure ($\tau_c^*$ and $\tau_s^*$ differ only in the sign before $\sqrt{\Psi}$), and do not involve any intermediate variables.
\end{remark}

\section{Concluding remarks}

In this paper, we studied the time–energy optimal control problem for double-integrator systems under state and control constraints. We provided a direct classification of admissible control profiles based on the boundary conditions, eliminating the need for sequential feasibility checks. We also simplified the structure of constrained trajectories. Specifically, for cases where both state and control constraints become active, we showed that the switching times can be computed via simple algebraic equations instead of cubic ones, reducing complexity. Furthermore, we derived closed-form expressions for all admissible profiles, including special pathological cases, enabling efficient and transparent trajectory generation. These results are particularly relevant for real-time applications such as CAVs, where vehicle onboard computational capacity is limited.
Future work should extend this framework to more complex dynamics and large-scale coordination settings.

\bibliographystyle{ieeetr}
\bibliography{bibliography,IDS_Publications_08212025}

\clearpage

\end{document}